\documentclass[a4paper]{article}

\usepackage{amsfonts}    
\usepackage{algorithm}
\usepackage[noend]{algorithmic}
\usepackage{geometry}
\usepackage{fullpage}
\usepackage{amsmath}
\usepackage{amsthm}
\usepackage{amssymb}
\usepackage{graphicx}
\usepackage{dsfont}
\usepackage{setspace}
\usepackage[shortlabels]{enumitem}
\setlist{nolistsep}

\newtheorem{theorem}{Theorem}
\newtheorem{lemma}[theorem]{Lemma}
\newtheorem{claim}{Claim}

\newtheorem{corollary}[theorem]{Corollary}

\newtheorem{observation}[theorem]{Observation}

\newenvironment{indentpar}[1]%
 {
 \begin{list}{}%
         {\setlength{\leftmargin}{#1}
         \setlist{nolistsep} }%
      \item[]%
 }
 {\end{list}}

\newcommand{\old}[1]{{}}
\newcommand{\versionA}[1]{{}}

\newcommand{\verAlgA}[1]{{}}

\newcommand{\bs}{\backslash}

\newcommand{\MCRA}{\textsc{MinRange}}
\newcommand{\MCRAS}{\textsc{MinRangeSpanner}}
\newcommand{\AlgMCRA}{\textsc{1DMinRA}}

\newcommand{\orderedline}{line alike}

\title{On the Minimum Cost Range Assignment Problem
}

\author{Paz Carmi\footnote{Department of Computer Science, Ben-Gurion University of the Negev, Israel}
 and Lilach Chaitman-Yerushalmi
\footnote{Department of Computer Science, Ben-Gurion University of the Negev, Israel}}

\begin{document}
\maketitle


\begin{abstract}
We study the problem of assigning transmission ranges to radio stations  
placed arbitrarily in a $d$-dimensional ($d$-D) Euclidean space in order to achieve a strongly connected communication network
with minimum total power consumption.
The power required for transmitting in range $r$ is proportional to $r^\alpha$,
where $\alpha$ is typically between $1$ and $6$, depending on various environmental factors.
While this problem can be solved optimally in $1$D, in higher dimensions it is known to be $NP$-hard for any $\alpha \geq 1$.

For the $1$D version of the problem, i.e., radio stations located on a line and $\alpha \geq 1$, 
we propose an optimal $O(n^2)$-time algorithm. 
This improves the running time of the best known algorithm by a factor of $n$.
Moreover, we show a polynomial-time algorithm for finding the minimum cost range assignment in $1$D
whose induced communication graph is a $t$-spanner, for any  $t \geq 1$.

In higher dimensions, finding the optimal range assignment is $NP$-hard;
however, it can be approximated within a constant factor. 
The best known approximation ratio is for the case $\alpha=1$, where the approximation ratio is $1.5$.
We show a new approximation algorithm with improved approximation ratio of $1.5-\epsilon$, where $\epsilon>0$ is a small constant.
\end{abstract}


\section{Introduction}\label{sec:Intro}
A wireless ad-hoc  network is a self-organized decentralized network that consists of independent radio transceivers (transmitter/receiver) 
and does not rely on any existing infrastructure. The network nodes (stations) communicate over radio channels. 
Each node broadcasts a signal over a fixed range and any node within this transmission range receives the signal.
Communication with nodes outside the transmission range is done using multi-hops, i.e., intermediate nodes pass the message forward 
and form a communication path from the source node to the desired target node.
The twenty-first century witnesses widespread deployment of wireless networks for professional and private
applications. The field of wireless communication continues to experience unprecedented market growth. 
For a comprehensive survey of this field see~\cite{pahlavan2005}.

Let $S$ be a set of points in the $d$-dimensional Euclidean space representing radio stations.
A \emph{range assignment} for $S$ is a function $\rho : S \rightarrow \mathds{R}^+$ that assigns each point a transmission range (radius). 
The cost of a range assignment, representing the power consumption of the network, is defined as 
$cost(\rho) =  \sum_{v \in S}(\rho(v))^\alpha$ for some real constant $\alpha \geq 1$, 
where $\alpha$ varies between $1$ and values higher than $6$, depending on different environmental factors~\cite{pahlavan2005}.

A range assignment $\rho$ induces a \emph{directed communication graph}
$G_{\rho}=(S,E_{\rho})$, where $E_{\rho}=\{(u,v):\rho(u) \geq |uv|\}$ and $|uv|$ denotes the Euclidean distance between $u$ and $v$. 
A range assignment $\rho$ is \emph{valid} if the induced (communication) graph $G_{\rho}$ is strongly connected.
For ease of presentation, throughout the paper we refer to the terms `assigning a range $|uv|$ to a point $u \in S$' 
and `adding a directed edge $(u,v)$' as equivalent.

We consider the $d$-D \textsc{Minimum Cost Range Assignment} (\MCRA) problem,
that takes as input a set $S$ of $n$ points in $\mathds{R}^d$,
and whose objective is finding a valid \emph{range assignment} for $S$
of minimum cost.
This problem has been considered extensively in various settings, for different values of $d$ and $\alpha$, 
with additional requirements and modifications.
Some of these works are mentioned in this section.

In \cite{Kirousis}, Kirousis et al. considered the $1$D \MCRA\ problem (the radio stations are placed
arbitrarily on a line) and showed an $O(n^4)$-time algorithm which computes an optimal solution for the problem.
Later, Das et al.~\cite{DasGN07} improved the running time to $O(n^3)$. Here, we propose an $O(n^2)$-time exact algorithm, 
this improves the running time of the best known algorithm by a factor of $n$ without increasing the space complexity. 
The novelty of our method lies in separating the range assignment into two, \emph{left} and \emph{right}, range assignments
(elaborated in Section~\ref{sec:Linear}). 
This counter intuitive approach allows us to achieve the aforementioned result and, moreover,
to compute an optimal range assignment in $1$D with 
the additional requirement that the induced graph is a $t$-spanner, for a given $t \geq 1$.

A directed graph $G=(S,E)$ is a $t$-spanner for a set $S$, if for every two points $u,v \in S$ 
there exists a path in $G$ from $u$ to $v$ of length at most $t|uv|$.
The importance of avoiding flooding the network when routing, 
was one of the reasons that led researchers to consider the combination of range assignment and $t$-spanners, 
e.g.,~\cite{KarimRCK09,Shpungin,WYLX02,WangLi}, 
as well as the combination of range assignment and hop-spanners, e.g.,~\cite{Clementi00,Kirousis}.
While bounded-hop spanners bound the number of intermediate nodes forwarding a message, $t$-spanners bound the relative distance a message is forward. 
For the $1$D bounded-hop range assignment problem, Clementi et al.~\cite{Clementi00} showed a 2-approximation algorithm 
whose running time is $O(hn^3)$.
To the best of our knowledge, we are the first to show an algorithm that computes an optimal solution for the 
range assignment with the additional requirement that the induced graph is a $t$-spanner.  

While the $1$D version of the \MCRA\ problem can be solved optimally,
for any $d \geq 2$ and $\alpha \geq 1$, it has been proven to be $NP$-hard 
(in~\cite{Kirousis} for $d\geq3$ and $1 \leq \alpha < 2$ and later in~\cite{Clementi} for $d \geq 2$ and $\alpha > 1$). 
However, some versions can be approximated within a constant factor. 
For $\alpha = 2$ and any $d \geq 2$
Kirousis et al.~\cite{Kirousis} gave a 2-approximation algorithm based on the minimum spanning tree 
(although they addressed the case of $d\in \{2,3\}$ their result holds for any $d \geq 2$). 
The best known approximation ratio is for the case $\alpha=1$, where the approximation ratio is $1.5$~\cite{Ambuhl05}.
We show a new approximation algorithm for this case\footnote{Values of $\alpha$ smaller than $2$ correspond to areas, 
such as, corridors and large open indoor areas~\cite{pahlavan2005}.} with improved approximation ratio of $1.5-\epsilon$, for a suitable constant $\epsilon>0$.
We do not focus on increasing $\epsilon$ but rather on showing that there exists an approximation ratio for this problem 
that is strictly less than $1.5$. This is in contrast to classic problems, such as metric TSP and strongly connected sub-graph problems, 
for which the $1.5$ ratio bound has not yet been breached.



\section{Minimum Cost Range Assignment in 1D}\label{sec:Linear}
In the $1$D version of the \MCRA\ problem, the input set $S=\{v_1,...,v_n\}$ consists of points located on a line. 
For simplicity, we assume that the line is horizontal and for every $i<j$, $v_i$ is to the left of $v_j$.
Given two indices $1 \leq i < j \leq n$, we denote by $S_{i,j}$ the subset $\{v_i,...,v_j\}\subseteq S$.

We present two polynomial-time algorithms for finding optimal range assignments, 
the first, in Section~\ref{sec:LinMinCostRange}, for the basic $1$D \MCRA \ problem,
and the second, in Section~\ref{subsec:spanner}, subject to the additional requirement 
that the induced graph is a $t$-spanner (the $1$D \MCRAS \ problem).
Our new approach for solving these problems requires introducing a variant of the \emph{range assignment}.
Instead of assigning each point in $S$ a radius,
we assign each point two directional ranges, 
\emph{left range assignment}, $\rho^l : S \rightarrow \mathds{R}^+$,
and \emph{right range assignment}, $\rho^r : S \rightarrow \mathds{R}^+$.
A pair of assignments $(\rho^l,\rho^r)$ is called a \emph{left-right assignment}.
Assigning a point $v \in S$ a left range $\rho^l(v)$ and a right range $\rho^r(v)$
implies that in the induced graph, $G_{\rho^{lr}}$, $v$ can reach every point to its left up to distance $\rho^l(v)$
and every point to its right up to distance $\rho^r(v)$.
That is, $G_{\rho^{lr}}$,
contains the directed edge $(v_i,v_j)$ iff one of the following holds:
(i) $i<j$ and $|v_iv_j| \leq \rho^r(v_i)$, or \ 
(ii) $j<i$ and $|v_iv_j| \leq \rho^l(v_i)$.
The cost of an assignment $(\rho^l,\rho^r)$, is defined as
$cost(\rho^l,\rho^r) =  \sum_{v \in S}(\max\{\rho^l(v),\rho^r(v)\})^\alpha$.

Our algorithms find a \emph{left-right assignment} of minimum cost
that can be converted into a \emph{range assignment} $\rho$ with the same cost
by assigning each point $v \in S$ a range $\rho(v)=\max\{\rho^l(v),\rho^r(v)\}$.
Note that any valid \emph{range assignment} for $S$ can be converted to a a \emph{left-right assignment}
with the same cost, by assigning every point $v \in S$, $\rho^l(v)=\rho^r(v)=\rho(v)$.
To be more precise, either $\rho^l(v)$ or $\rho^r(v)$ should be reduced to $|vu|$, 
where $u$ is the farthest point in the directional range 
(for Lemma~\ref{lem:range_ij} to hold).
Therefore, a minimum cost \emph{left-right assignment}, implies a minimum cost \emph{range assignment}.

In addition to the $cost$ function, we define 
$cost'(\rho^l,\rho^r) =  \sum_{v \in S}((\rho^l(v))^\alpha+(\rho^r(v))^\alpha),$
and refine the term of \emph{optimal solution} to include
only solutions that minimize $cost'(\rho^l,\rho^r)$
among all solutions, $(\rho^l,\rho^r)$, with minimum $cost(\rho^l,\rho^r)$.


\subsection{An Optimal Algorithm for the 1D \MCRA \ Problem}\label{sec:LinMinCostRange}
Das et al.~\cite{DasGN07} state three basic lemmas regarding properties of an optimal range assignment.
The following three lemmas are adjusted versions of these lemmas for a \emph{left-right assignment}.
\begin{lemma}\label{lem:range_ij}
In an optimal solution $(\rho^l, \rho^r)$ 
for every $v_i \in S$,
either  $\rho^l(v_i)=0$ or $\rho^l(v_i)=|v_i v_j|$ 
and similarly, either $\rho^r(v_i)=0$ or $\rho^r(v_i)=|v_i v_k|$ for some $j \leq i \leq k$.
\end{lemma}
\begin{lemma}\label{lem:range_line}
Given three indices $1\leq i < j < k \leq n$, consider an optimal solution for $S_{i,k}$, 
denoted by $(\rho^l, \rho^r)$, subject to the condition that $\rho^l(v_j)\geq |v_i v_j|$ and $\rho^r(v_j) \geq |v_j v_k|$,
then, 
\begin{itemize}
	\item  for all $m=i,...,j-1$, $\rho^r(v_m)=|v_m v_{m+1}|$ and $\rho^l(v_m)=0$; and
   \item for all $m=j+1,...,k$, $\rho^l(v_m)=|v_m v_{m-1}|$ and $\rho^r(v_m)=0$.
\end{itemize}
%
\end{lemma}
%
%
\begin{lemma}\label{lem:range_12}
In an optimal solution $(\rho^l, \rho^r)$,
$\rho^l(v_1)=0$ and $\rho^r(v_1)=|v_1v_2|$.
\end{lemma}

Lemma~\ref{lem:range_ij} allows us to simplify the notation 
$\rho^x(v_i)=|v_i v_j|$ for $x \in \{l,r\}$ and $1 \leq i,j \leq n$,  
and write $\rho^x(i)=j$ for short.
We solve the \MCRA\ problem using dynamic programming.
Given $1 \leq i <n$, we denote by $OPT(i)$ the cost of an optimal solution for the sub-problem
defined by the input $S_{i,n}$, subject to the condition that $\rho^r(i)=i+1$.
Note that the cost of an optimal solution for the whole problem is $OPT(1)$.

In Section~\ref{subsec:cubic} we present an algorithm with $O(n^3)$ running time and $O(n^2)$ space
(the same time and space as in~\cite{DasGN07}).
Then, in Section~\ref{subsec:Quad} we 
reduce the running time to $O(n^2)$.


\subsubsection{A Cubic-Time Algorithm 
}\label{subsec:cubic}
Algorithm~\AlgMCRA\ (Algorithm~\ref{alg}) applies dynamic programming to compute the values $OPT(i)$
for every $1\leq i \leq n$ and store them in a table, $T$.
Finally, it outputs the value $T[1]$.
In our computation we use a $2$-dimensional matrix, \emph{Sum}, storing 
for every $1\leq i<j \leq n$ the sum $\sum_{m=i}^{j-1} |v_m v_{m+1}|^\alpha$.
\begin{algorithm}[htp]
\caption{\AlgMCRA($S$)}\label{alg}
\begin{algorithmic}
	\FOR {$i = n-1$ \textbf{downto} $1$}
			\FOR {$j = n$  \textbf{downto}  $i+1$}
				\STATE	
				\[
							\text{Sum}[i,j] \gets 
										\left\{
											\begin{array}{ll}  
								        						 |v_{n-1} v_n|^\alpha &\:  ,i=n-1 \\
								        				\text{Sum} [i+1,n]+|v_i v_{i+1}|^\alpha &\:  ,j=n \\
								        				\text{Sum} [i,j+1]-|v_j v_{j+1}|^\alpha &\:  ,\text{otherwise}	
											\end{array}
										\right.
					\]
				\ENDFOR
		\ENDFOR	
		\FOR {$i = n-1$ \textbf{downto} $1$}
		\STATE
				\[
							\text{T}[i] \gets  
										\left\{
											\begin{array}{cl}  
								        				\text 2|v_i v_{i+1}|^\alpha &\:  ,i=n-1 \\
								        				\begin{aligned}
	        												\min_{ \substack{i < k < n \\ k <  k' \leq n}} 
																	\{&  \text{Sum}[i,k'-1]  + \text{T}[k'-1]  - |v_{k'-1} v_{k'}|^\alpha\\[-2mm]
																 	&  + \max\{ |v_i v_k|^\alpha, |v_k v_{k'}|^\alpha \}\}
															 \end{aligned}  &\: ,otherwise 
											\end{array}
										\right.
					\]
	\ENDFOR
	\RETURN $\text{T}[1]$
\end{algorithmic}
\end{algorithm}
While the table $T$ maintains only the costs of the solutions,
the optimal assignment can be easily retrieved by 
backtracking the cells leaded to the optimal cost and assigning the associated ranges (described in the proof of Lemma~\ref{lem:OPT(i)}).\\

\paragraph*{Correctness.}
%
We prove that for every $1 \leq i \leq n$, the value assigned to cell $T[i]$ by the algorithm 
equals $OPT(i)$.
Trivially, $OPT(n-1)$ indeed equals $2|v_i v_{i+1}|^\alpha$.
Assume, during the $i$-th iteration it holds that $T[i']=OPT(i')$ for every $i<i'<n$,
the correctness of the computation done during the $i$-th iteration is given in Lemma~\ref{lem:OPT(i)}.
\begin{lemma}\label{lem:OPT(i)}
Given an index $i$ with $1 \leq i \leq n-1$,	
\[ OPT(i) = \min_{ \substack{i < k < n \\ k <  k' \leq n}}  
						\left\{	\sum_{m=i}^{k'-2}  |v_m v_{m+1}|^\alpha  + OPT(k'-1)  
								 - |v_{k'-1} v_{k'}|^\alpha + \max\{ |v_i v_k|^\alpha, |v_k v_{k'}|^\alpha \}	\right\}.		\]			
\end{lemma}
\begin{proof}
Let $X_i$ denote the right side of the equation, we prove $OPT(i) = X_i$. 
%
\begin{description}[topsep=2mm, itemsep=2mm, leftmargin=2.5mm]
\item [$OPT(i) \leq X_i$:]  \
We show that all costs that appear as $\min$ function arguments  in $X_i$
correspond to valid assignments and thus infer, by the optimality of $OPT(i)$, that the above inequality holds.
Consider an argument with parameters $k$ and $k'$. 
We associate it with an assignment $(\rho^l,\rho^r)$ defined as follows (see Fig.~\ref{fig:OPT(i)}(a)).
For $m \geq k'-1 $ the assignment is inductively defined by $OPT(k'-1)$.
For every $i \leq m < k$, $\rho^l(m)=m$ and $\rho^r(m)=m+1$,
for every $k < m < k'$, $\rho^l(m)=m-1$ and $\rho^r(m)=m$ ($v_{k'-1}$ is reassigned) and 
for $k$, $\rho^l(k)=i$, $\rho^r(k)=k'$.
By the validity of $OPT(k'-1)$, every two points among $S_{k'-1,n}$ are (strongly) connected.
Our assignment for $S_{i,k'-1}$ guarantees the connectivity between every two points
in $S_{i,k'-1}$, and thus between every two points in $S_{i,n}$.
\item[$OPT(i) \geq X_i$:] \
Consider an optimal solution $(\rho^l,\rho^r)$ for the points $S_{i,n}$ 
subject to the condition that $\rho^r(i)=i+1$. 
Let $v_k$ be a point to the right of $v_i$ with $\rho^l(k)=i$ and let $\rho^r(k)=k'$.
Note that since $v_i$ is the leftmost point and the induced graph
is strongly connected, such a point necessarily exists.

Next we show that there is no edge directed either left or right 
connecting two points on different sides of $v_{k'}$ in $G_{\rho^{lr}}$, except for possibly an edge $(v_j, v_{k'-1})$ with $j>k'$.
Assume towards contradiction that the former does not hold, i.e., there exists $i<t<k'$, with $\rho^r(t)\geq k'$;
then, reassigning $\rho^r(k)= \max \{t,k\}$ maintains the connectivity,
and reduces the value of $cost'$ without increasing the value of $cost$
in contradiction to the optimality of the solution.
Now, let $v_j$ be a point to the right of $v_{k'}$ with $\rho^l(j)=j'\in \left[ i,k' \right]$,
we show that $j' \geq k'-1$.
Consider a point $j'<t<k'$, as we have shown, $\rho^r(t)<k'$.
By symmetric arguments we have $\rho^l(t)> j'$  (see Fig.~\ref{fig:OPT(i)}(b)). 
Namely, there is no edge going out of the interval $(v_{j'},v_{k'})$.
Thus, connectivity can be achieved only if this interval is empty of vertices, 
i.e., either $j'=k'-1$ or $j'=k'$ (note that $k'-1>i$).  
 
The above observation allows us to divide the problem into two independent subproblems,
one for the points $S_{i,k'-1}$ subject to the constraints $\rho^l(k)=i$ and $\rho^r(k)=k'$,
and the other for the points $S_{k'-1,n}$ subject to the artificial constraint $\rho^r(k'-1)=k'$ 
that guarantees the existence of a path from $k'-1$ to $k'$, due to the solution of the first subproblem,
but should not be paid for.
Regarding the first subproblem, by Lemma~\ref{lem:range_line}, in an optimal assignment,
for every $i \leq m < k$, $\rho^l(m)=m$ and $\rho^r(m)=m+1$, and
for every $k < m \leq k'-1$, $\rho^l(m)=m-1$ and $\rho^r(m)=m$.
Thus, its cost is $ \sum_{m=i}^{k'-2} |v_m v_{m+1}|^\alpha + \max \{|v_k v_i|^\alpha ,|v_k v_{k'}|^\alpha$\}.
The cost of an optimal solution to the second subproblem is $OPT(k'-1) - |v_{k'-1} v_k'|^\alpha$.
Hence, the $cost$ of an optimal solution to the whole problem is the sum of the above costs and the lemma follows.
\end{description}
\end{proof}
\begin{figure}[tbh]
    \centering
        \includegraphics[width=1\textwidth]{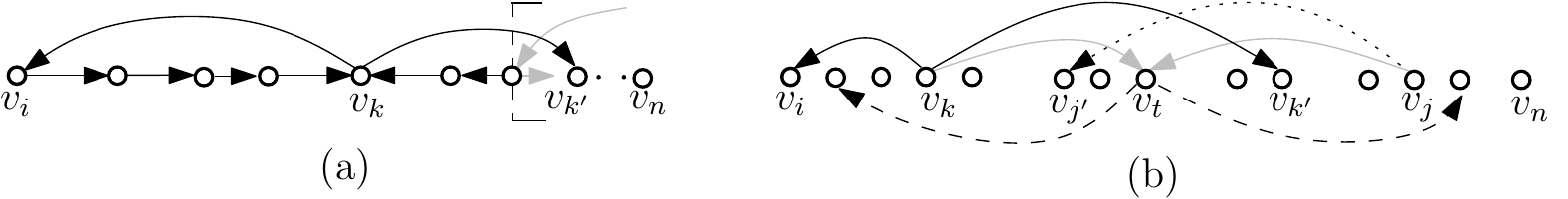}
    \caption{(a) An illustration of the assignment associated with $OPT(i)$ with respect to given 
    parameters $k$ and $k'$. In gray are range assignments associated with $OPT(k'-1)$. 
    (b) An illustration of the proof of Lemma~\ref{lem:OPT(i)}. 
    In dashed arrows, the impossible ranges of $v_t$ and in gray, the alternative assignment of lower $cost'$.}
    \label{fig:OPT(i)}
\end{figure}
%
%
\paragraph*{Complexity.}
%
Obviously, Algorithm~\AlgMCRA\ requires $O(n^2)$ space. 
Regarding the running time, $O(n)$ iterations are performed during the algorithm,
each iteration takes $O(n^2)$ time 
Therefore, the total running time 
is $O(n^3)$ and Lemma~\ref{lem:alg1DRA} follows.
\begin{lemma}\label{lem:alg1DRA}
Algorithm~\AlgMCRA\ runs in $O(n^3)$ time using $O(n^2)$ space. 
\end{lemma}
%

\subsubsection{A Quadratic-Time Algorithm}\label{subsec:Quad}
In this section we consider Algorithm~\AlgMCRA\ from previous section
and reduce its running time to $O(n^2)$.
Consider the equality stated in Lemma~\ref{lem:OPT(i)}.
Observe that given fixed values $i$ and $k'$, the value $k$ that minimizes the argument of the $\min$ function
with respect to $i$ and $k'$ is simply the value $k$ that minimizes $\max\{ |v_i v_k|^\alpha, |v_k v_{k'}|^\alpha\}$. 
This value is simply the closest point to the midpoint of the segment $\overline{v_i v_{k'}}$,
denoted by $c(i,k')$. Thus, 
\[ 	OPT(i)=		
	        				\min_{i+1 < k' \leq n} 
										\left\{
										\begin{aligned}
										 \sum_{m=i}^{k'-2}  |v_m v_{m+1}|^\alpha  &+ OPT(k'-1) - |v_{k'-1} v_{k'}|^\alpha \\
									 	& + \max\{ |v_i v_{c(i,k')}|^\alpha, |v_{c(i,k')} v_{k'}|^\alpha \}
									 	 \end{aligned} 
									 	\right\}.				 	
\]
Consider Algorithm~\AlgMCRA\ after applying the above modification 
in the computation of $T[i]$.
Since there are only $O(n)$ sub-problems to compute, each in $O(n)$ time,
the running time reduces to $O(n^2)$ and the following theorem follows.
\begin{theorem}
The $1$D \MCRA\ problem can be solved in $O(n^2)$ time using $O(n^2)$ space.
\end{theorem}
%
%
\subsection{An Optimal Algorithm for the 1D \MCRAS \ Problem}\label{subsec:spanner}
 
Given a set $S=\{v_1,..,v_n\}$ of points in $1D$ and a value $t \geq 1$,
the 1D \MCRAS\ problem aims to find a minimum cost range assignment 
for $S$, subject to the requirement that the induced graph is a $t$-spanner.
We present a polynomial-time algorithm which solves this problem optimally 
and follows the same guidelines as Algorithm~\AlgMCRA. 

We begin with providing the key notions required for understanding 
the correctness of the algorithm, followed by its description. 
Due to space limitation, we do not supply a formal proof. 
The first and most significant observation, is that the problem can still be divided into two subproblems
in the same way as in Algorithm~\AlgMCRA, by similar arguments to those of Lemma~\ref{lem:OPT(i)}.
In Lemma~\ref{lem:OPT(i)} we show that any assignment that does not satisfy the conditions required for the division
can be adjusted to a new assignment with a lower value of $cost'$ that preserves connectivity.
The new assignment, however, preserves also the lengths of the shortest paths,
which make the argument legitimate for this problem as well. 

The two problems (\MCRA\ and \MCRAS\ ) differ when it comes to solving each of the above subproblems.
Consider the left subproblem, i.e., of the form described in Lemma~\ref{lem:range_line}.
The optimal assignment for it is no longer necessarily the one stated in the lemma,
since it does not ensure the existence of $t$-spanning paths.
Therefore, our algorithm divides problems of this form into smaller subproblems handled recursively (see Fig.~\ref{fig:range_spanner}, right). 
Dealing with such subproblems, requires defining new parameters: a rightmost input point $v_j$, and
the length of the shortest paths connecting $v_i$ to $v_j$, $v_j$ to $v_i$ and $v_i$ to $v_{i+1}$ 
not involving points in $S_{i,j}$ except for the endpoints, 
denoted by $\overrightarrow{\delta}$, $\overleftarrow{\delta}$, and $\delta^i$ 
, respectively.
Regarding the computation of a subproblem, since points may be covered now by vertices outside the subproblem domain,
we allow $v_k$ to have either a right or a left range equals $0$ (in the terms of Algorithm~\AlgMCRA, either $k=i$ or $k=k'$).

Another key observation is that any directed graph $G$ over $S$
is a $t$-spanner for $S$ iff for every $1 \leq i <n$ there exists a $t$-spanning path from
$v_i$ to $v_{i+1}$ and from $v_{i+1}$ to $v_i$. 
Moreover, given that $G$ is strongly connected implies that the addition of an edge between consecutive points
does not effect the length of the shortest path between any other pair of consecutive points.
Therefore, for subproblems with $j=i+1$ we assign $\rho^r(i)=i+1$ (resp. $\rho^l(i+1)=i$) iff 
$\overrightarrow{\delta}/|i,i+1|>t$ (resp. $\overleftarrow{\delta}/|i,i+1|>t$) and thus ensuring that the induced graph is a $t$-spanner.

Our algorithm may consider solutions in which an assignment to a node is charged more than once in the total cost;
however, for every such solution, there exists an equivalent one in which the charging is done properly and
is preferred by the algorithm due its lower cost.

We denote by $OPT(i,j,\overrightarrow{\delta},\overleftarrow{\delta}, \delta^i)$ the cost of an optimal solution
to the sub-problem defined by the input $S_{i,j}$ subject to the parameters 
$\overrightarrow{\delta},\overleftarrow{\delta}$, and $\delta^i$ representing
the lengths of the shortest external paths as defined earlier in this section.
Let $\Delta_{i,j}=\{2|v_l v_i|+|v_i v_j|, 2|v_j v_r|+|v_i v_j| : l \leq i < j \leq r\}$.
We compute $OPT(i,j,\overrightarrow{\delta},\overleftarrow{\delta}, \delta^i)$  
for every $1\leq i<j \leq n$, and the corresponding $\overrightarrow{\delta},\overleftarrow{\delta}, \delta^i \in \Delta_{i,j} $ iteratively, 
while in stage $x$ all subproblems with $j-i=x$ are solved.
The computation is derived from the equalities below.
For simplicity of presentation, we overload notation and write $|i,j|$ to mean $|v_i v_j|$.
In addition, we write $\varnothing$ in place of $\delta^i$ where $\delta^i=\overrightarrow{\delta}$.
\[ 	
\begin{aligned}
	OPT(&i, i+1,\overrightarrow{\delta}, \overleftarrow{\delta}, \delta^i)= \overrightarrow{r} + \overleftarrow{r}, \text{where}\\
		  &\overrightarrow{r} =   	
								\left\{
											\begin{array}{ll}  
								        			 |i,i+1| \text{ (*\emph{assigning $\rho^r(i)=i+1$}*)}&\:  ,\overrightarrow{\delta}/|i,i+1|> t \\
															 0 &\: ,otherwise 
											\end{array}
								\right.
\end{aligned}
\]
and  $\overleftarrow{r}$ is defined symmetrically.
\noindent
For $j>1$, we have
\[ 	
\begin{aligned}
	&OPT(i, j,\overrightarrow{\delta}, \overleftarrow{\delta}, \delta^i)= \\
	&\min_{\substack{i \leq k \leq j \\ k \leq k' \leq j}}   	
					\left\{
						\begin{array}{ll}
							\begin{aligned}  
								    |i,i+1|^\alpha &+ OPT(i,\, i+1, \,|i,i+1|, \,|i+i,j|+\overleftarrow{\delta}, \,\varnothing) \\
								        					&+ OPT(i+1, \,j, \,\infty, \,\overleftarrow{\delta}-|i+1, i|, \,\varnothing) 
								  \end{aligned} &\:  ,i=k \\[4mm]
								\begin{aligned}  
								    |i,i+1|^\alpha &+ OPT(i\,i+1,\,\delta^i,\,|i,i+1|,\,\varnothing) \\
								        					&+ OPT(i+1,\,j,\,|i,i+1|+\overrightarrow{\delta},\,\overleftarrow{\delta}-|i,i+1|,\,\varnothing) 
								 \end{aligned} &\:  ,k=k' \\[4mm]
								  \begin{aligned}  
								    &\max\{|i,k|^\alpha, |k,k'|^\alpha\} \\
								    																	&+ OPT(i,\,i+1,\,\delta^i,\,|i+1,k|+|k,i|,\,\varnothing) \\
								        															&+ OPT(i+1,\,k,\,\infty,\,|k,i+1|,\,\infty) \\
								        															&+ OPT(k,\,k'-1,\,|k,k'-1|,\,\infty,\,|k,k+1|) \\
								        															&+ OPT(k'-1,\,j,\,|k'-1,i|+\overrightarrow{\delta},\,\overleftarrow{\delta}-|i,k'-1|,\,|k'-1,k|+|k,k'|). \\
								  \end{aligned} &\:  ,i \neq k \neq k' \\
							\end{array}
					\right.\\
\end{aligned}
\]
We permit either $i=k$ and then $k'=i+1$ or $k=k'$ and then $k'=i+1$ but not both. 
%
\begin{figure}[thb]
    \centering
        \includegraphics[width=1\textwidth]{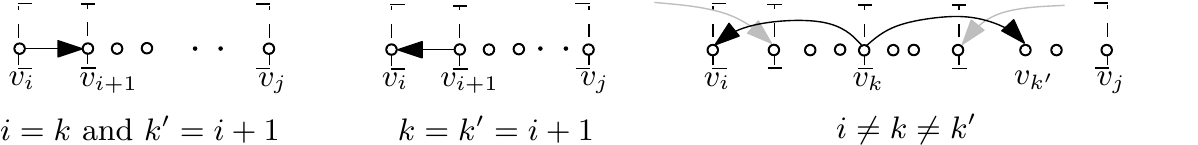}
    \caption{An illustration of the algorithm for the \MCRAS\ problem. 
    The ranges are illustrated in black arrows and 
    the division to subproblems in sashed lines.}
    \label{fig:range_spanner}
\end{figure}
\paragraph*{Complexity.}
Let $\Delta$ be the set of all distinct distances in $S$, then for every $v_i,v_j \in S $, 
$|\Delta_{i,j}|=|\Delta|=O(n)$.
We fill a table with $O(n^2 |\Delta|^3)$ cells, each cell is computed in $O(n^2)$ time,
thus, the total running time is $O(n^4 |\Delta|^3)$.
As we have focused on presenting a simple and intuitive solution, rather than reducing the running time,
a more careful analysis achieves a better bound on the time complexity.
For example, the relevant domain of $\overrightarrow{\delta},\overleftarrow{\delta}$, and $\delta^i$
can be estimated more precisely with respect to $t$.
Moreover, Observation~\ref{obs:lr_ranges} allows reducing the running time by a factor of $n$.
This is done by decreasing the number of relevant combinations of $i$ and $k'$
that have to be checked by the algorithm, for fixed indices $j$ and $k$ with $i < k < k'< j $, to $O(n)$,
using similar arguments to those in Lemma~\ref{lem:OPT(i)} (see Fig.~\ref{fig:OPT(i)_improved}).
\begin{observation}\label{obs:lr_ranges}
Consider an optimal assignment $(\rho^l,\rho^r)$ and a point $v_k \in S$.
Let $\rho^l(k)=i$ and let $k^{i}$ denote the minimal index with $k < k^{i}$ and $|v_k v_{k^{i}}|\geq|v_k v_{i}|$,
then $k^{i+1} \leq \rho^r(k) \leq k^{i}$.
\end{observation}
\begin{figure}[bht]
    \centering
        \includegraphics[width=0.75\textwidth]{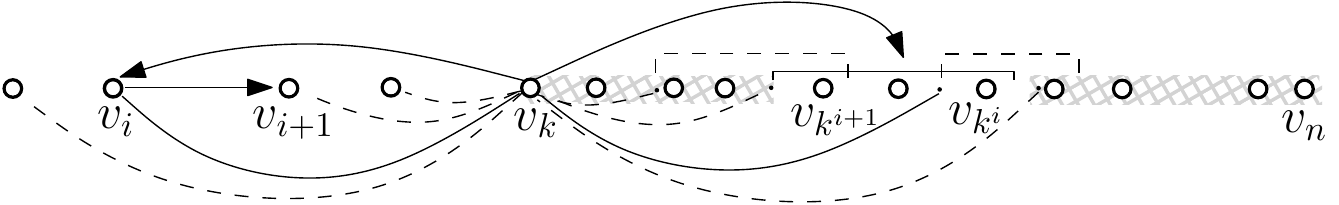}
    \caption{An illustration of Observation~\ref{obs:lr_ranges}.
    				Every pair of symmetric arcs indicates equal distances from $v_k$. 
    				The marked domains indicates the legal values of
    				$\rho^r(k)$ for different values of $i$.}
    \label{fig:OPT(i)_improved}
\end{figure}


\section{The \MCRA\ Problem in Higher Dimensions}\label{sec:PlaneMinCostRange }

In this section we focus on the \MCRA\ problem
for dimension $d \geq 2$ and $\alpha =1$.
As all the versions of the problem for  $d \geq 2$ and $\alpha \geq 1$, it is known to be $NP$-hard.
Currently, the algorithm achieving the best approximation ratio for $\alpha=1$ and ant $d\geq 2$ is the \emph{Hub} algorithm 
with a ratio of $1.5$.
This algorithm was proposed by G. Calinescu, P.J.Wan, and F. Zaragoza for the general metric case, and
analyzed by Amb\"{u}hl et al. in~\cite{Ambuhl05} for the restricted Euclidean case. 
We show a new approximation algorithm and bound its approximation ratio from above 
by $1.5-\epsilon$ for $\epsilon=5/10^5$. Although in some cases our phrasing is restricted to the plane, 
all arguments hold for higher dimensions as well.

In our algorithm we use two existing algorithms, 
the \emph{Hub} algorithm and the algorithm for the $1$D \MCRA\ problem introduced by Kirousis et al.~\cite{Kirousis}, 
to which we refer as the \emph{$1$D RA} algorithm.
We observe that the later algorithm outputs an optimal solution for any ordered set $V=\{v_1,...,v_n\}$ with distance function $h$
that satisfies the following \emph{\orderedline\ } condition:
for every $1 \leq i \leq j < k \leq l \leq n$,
it holds that $h(v_i,v_l) \geq h(v_j,v_k)$. 
We use this algorithm for subsets of the input set that roughly lie on a line.


\subsection{Our Approach}\label{sec:approach}
Presenting our approach requires acquaintance with the \emph{Hub} algorithm.
The \emph{Hub} algorithm finds the minimum enclosing disk $C$ of $S$ centered at point $hub \in S$. 
Then, the algorithm sets $\rho(hub)=r_{min}$, where $r_{min}$ is $C$'s radius.  Finally, 
it directs the edges of $MST(S)$ towards the \emph{hub} and for each directed edge $(v,u)$ sets $\rho(v)=|vu|$.
The cost of this assignment is 
$w(MST(S))+r_{min} \leq w(MST(S))+(w(MST(S))+w(e_{M}))/2,$
where $e_M$ is the longest edge in $MST(S)$ and the weight function $w$ is defined with respect to Euclidean lengths.

To guide the reader, we give an intuition and a rough sketch of our algorithm.
We characterize the instances where the \emph{Hub} algorithm gives a better approximation than 1.5, 
and to generalize these cases we slightly modify it. 
Furthermore, we show an algorithm that prevails in the cases where the modified \emph{Hub} algorithm 
fails to give an approximation ratio lower than 1.5.
Before we elaborate more on the aforementioned characterization,
another piece of terminology. 
Given a graph $G$ over $S$ and two points $p,q \in S$, the \emph{stretch factor} from $p$ to $q$ in $G$
is $\delta_G(p,q)/|pq|$,
where $\delta_G(p,q)$ denotes the Euclidean length of the shortest path between $p$ and $q$ in $G$.
We use $\sim$$large$ when referring to values greater than fixed thresholds, some with respect to $w(MST(S))$, defined later.

Consider $MST(S)$ and its longest path $P_M$. If one of the following conditions holds, then 
the \emph{Hub} algorithm or its modification results in a better constant approximation than $1.5$: \ 
	(A1) there exists a $\sim$$large$ edge in $MST(S)$; \
	(A2) a $\sim$$large$ fraction of $P_M$ consists of disjoint sub-paths connecting pairs of points with $\sim$$large$ stretch factor,
			not dominated by one sub-path of at least half the fraction; or \
  (A3) the weight $w(MST(S)\bs P_M)$ is $\sim$$large$.

Otherwise, there are three possible cases: \
  (B1) the graph $MST(S)$ is roughly a line; \
	(B2) there are two points in $P_M$ with $\sim$$large$ stretch factor,
		i.e., there is a $\sim$$large$ `hill' in $P_M$, and then either $MST(S)$ roughly consists of two  $1$D paths; or \
  (B3) the optimal solution uses edges connecting the two sides of the `hill', covering $\sim$$large$ fraction of it.

The last three cases are approximated using the following method.
We consider every two edges connecting the two sides of the `hill' 
as the edges in the optimal solution that separates the uncovered remains of the path to two independent sub-paths, i.e., not connected by an edge.
(the points in both sub-paths may be connected to the middle covered area.)
Note that such two edges exist.
We direct the covered area to achieve a strongly connected sub-graph and solve each of the two sub-paths separately in two techniques.
The first, using the \emph{$1$D RA} algorithm with a distance function implied by the input, satisfying the \emph{\orderedline\ } condition, 
and applying adjustments on the output,
and the second, using the \emph{Hub} algorithm.
A $(1.5-\epsilon)$-approximation is obtained for cases (B1) and (B2), using the first technique,
and for cases (B1) and (B2), using the second technique.
The algorithm computed several solutions, using the aforementioned methods, and returns the one of minimum cost.

\old{
\begin{figure}[htb]
    \centering
        \includegraphics[width=0.78\textwidth]{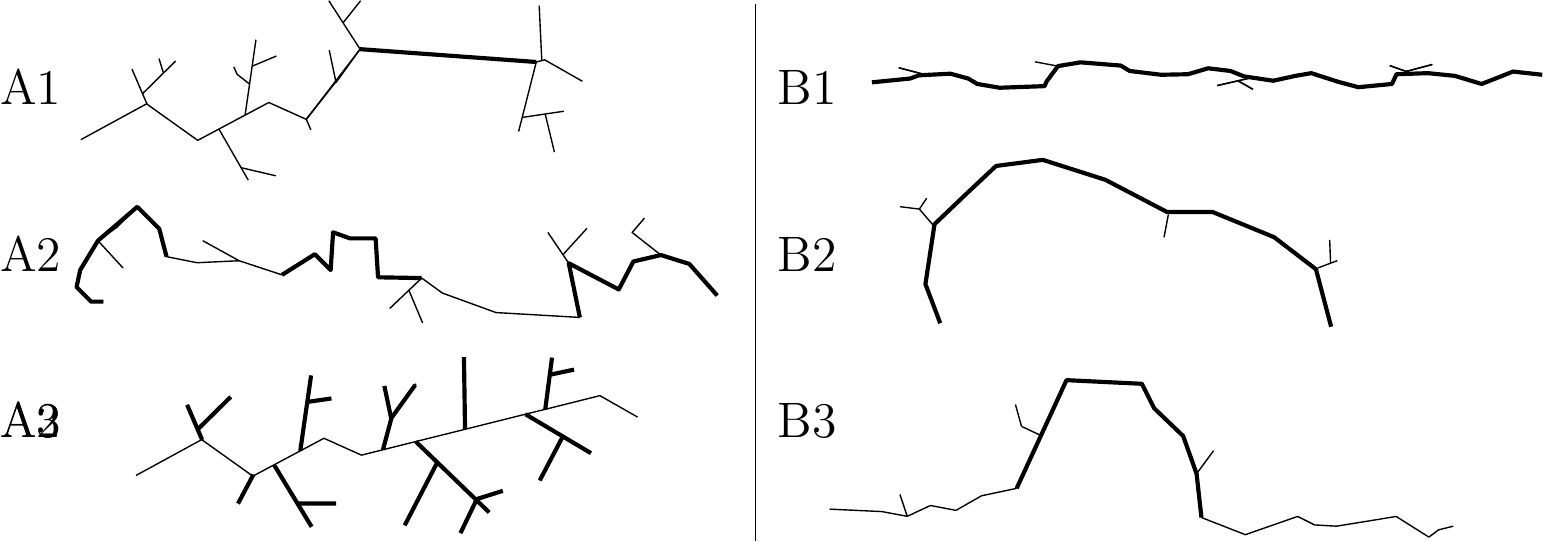}
    \caption{Left, the cases where the required approximation is achieved by the \emph{Hub} algorithm;
    right, the cases where the required approximation is achieved by other methods.}
    \label{fig:approx_approach}
\end{figure}
}


\subsection{The Approximation Algorithm}\label{sec:alg}
The algorithm uses the following three procedures that are defined precisely at the end of the algorithm's description.
\begin{itemize}[leftmargin=*]
\item The \emph{flatten} procedure $f$ - a method performing shortcuts between pairs of points 
			on a given path $P$ resulting in a path without two points of stretch factor greater than $c_s$.
\item The \emph{distance function} $h_S$ - a distance function defined for an ordered set $P\subseteq S$, satisfying the \emph{\orderedline} condition. 
\item The \emph{adjustment} transformation $g$ - a function adjusting an optimal range assignment for an ordered set $P\subseteq S$ with distance function $h$,
			to a valid assignment for $P$.
\end{itemize}
\vspace{0.5mm}
Let $R$ be the forest obtained by omitting from $MST(S$) the edges of its longest path, $PM$.
Given a point $v \in P_M$, 
let $T(v)$ denote the tree of $R$ rooted at $v$.
For every $u \in T(v)$ let $r(u)$ denote the root of the tree in $R$ containing $u$, namely, $v$.
For a set of points $V \subset P_M$, let $T(V)$ denote the union $\bigcup_{v \in V} T(v)$.
For ease of presentation, we assume the path $P_M$ has a \emph{left} and a \emph{right} endpoints, 
thus, the \emph{left} and \emph{right} relations over $P_M$ are naturally defined. \\[2mm]
\textbf{The main algorithm scheme:} \\
%
Compute four solutions and return the one of minimal cost.
In case of multiple assignments to a point in a solution, the maximal among the ranges counts.\\
\underline{Solution} (i): apply the \emph{Hub} algorithm.\\
\underline{Solution} (ii): apply a variant of the \emph{Hub} algorithm - find a point $c \in P_M$ that minimizes the value $r_c=\max\{|cp_1|,|cp_z|\}$, 
where $p_1$ and $p_z$ are the endpoints of the path $P_M$.
Assign $c$ the range $r_c$, direct $P_M$ towards $r_c$ and bi-direct all edges in $R$.\\[1mm]
(* \textit{The rest of the algorithm handles cases (B1)-(B3) defined in Section~\ref{sec:approach}} *)\\[1mm]
\textbf{For} every edge $e \in P_M$ do :
\begin{indentpar}{0.4cm}
			Let $P_{e^l}$ and $P_{e^r}$ be the two paths of $P_M\bs e$, to the left and to the right of $e$, respectively.\\
			Apply the \emph{flatten procedure} $f$ on $P_{e^l}$ and $P_{e^r}$ to obtain the sub-paths\\
			$P_{l'}=(p_{1},p_{2},...,p_{m})$ and $P_{r'}=(p_{m+1},p_{i+2},...,p_z)$, respectively.\\ 
			(* \textit{Note $R$ has been changed during the \emph{flatten procedure}} *)\\
			\textbf{For} every $4$ points $p_{l}, p_{l'}, p_{r'}, p_{r}$
			with $l \leq l' \leq m <  r' \leq r$ in the flattened sub-paths:
			\begin{indentpar}{0.6cm}
	          In both solutions (iii) and (iv) direct the path $P_x=(p_l,...p_m,p_{m+1},...,p_r)$ towards $p_{l}$ 
		        and for each point $p_i$ with $1 \leq i \leq z$ direct $T(p_i)$ towards $p_i$ and assign $p_i$ a range $w(T(p_i))$.		
						Perform the least cost option among the following two, either 
						add the edge $(p_l,p_r)$, or add the two edges, 
						one from $u_l$ to $u_{r'}$ for $u_l \in T(p_l), u_{r'} \in T(p_{r'})$ of minimal length
						and the other from $u_{l'}$ to $u_r$ for $u_{l'} \in T(p_{l'}), u_r \in T(p_r)$ of minimal length.
						(see illustration in Fig.~\ref{fig:alg2}).
						As for the two sub-paths $P_l=(p_{1},p_{2},...,p_{l})$ and $P_r=(p_{r},p_{r+1},...,p_z)$,
						assign them ranges as follows:\vspace{0.1cm} \\
						\underline{Solution} (iii): apply the \emph{Hub} algorithm separately on each sub-path.\\ 
						\underline{Solution} (iv): apply the \emph{$1$D RA} algorithm separately on each sub-path 
						with respect to the \emph{distance function $h_S$}
						and perform the \emph{transformation $g$} on the assignment received.
			\end{indentpar}
\end{indentpar}
 
\begin{figure}[bht]
    \centering
        \includegraphics[width=0.6\textwidth]{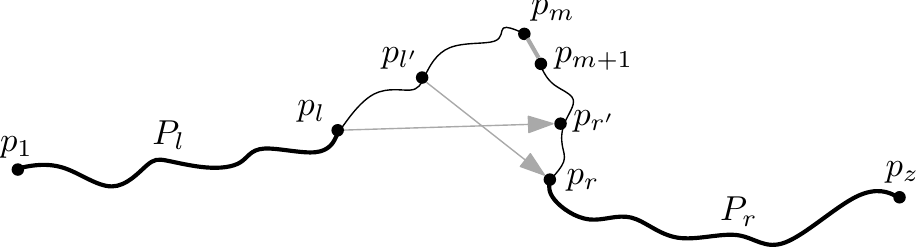}
    \caption{The two sub-paths $P_l$ and $P_r$ as defined in the algorithm.}
    \label{fig:alg2}
\end{figure}

\begin{description}[topsep=0.2cm, itemsep=0.1cm]

\item[The \emph{flatten} procedure $f$.]
Let $c_s = 5/4$. Given a path $P=\{v_i, ..,v_n \}$, set $Q_{P} = \{\}$. 
Let $j>i$ be the maximal index such that $\delta_{P}(v_i,v_j) > c_s|v_i v_j|$.
If such index does not exist, let $j=i+1$.
Else ($j>i+1$), add the edge $(v_i,v_j)$ to $P$, remove the edge $(v_{j-1},v_j)$ from $P$, 
move the sub-path $(v_i,..,v_{j-1})$ from $P$ to the forest $R$, 
and update $Q_{P} = Q_{P} \cup \{ (v_i,v_j)\}$. \
Finally, repeat with the sub-path $(v_j,..,v_n)$ without initializing $Q_P$.
\end{description}


The definitions for $h_S$ and $g$ are given with respect to
the sub-paths $P_l$ and $P_{l'}$, the definitions for the sub-path $P_r$ and $P_{r'}$ are symmetric.
\begin{description}[topsep=0.2cm, itemsep=0.1cm]
\item[The distance function $h_S$.]
For every two points $p_j,p_k$ with $1\leq j \leq k \leq l$ we define,
$$h_S(p_j,p_k) = \min_{ \substack{u \in T(p_{j'}), 1\leq j' \leq j \\ 														  	
																v \in T(p_{k'}), k\leq k' \leq m}}		|uv|.$$

\item[The \emph{adjustment} transformation $g$.]
Given an assignment $\rho':P_l \rightarrow \mathds{R}^+$, 
we transform it into an assignment $g(\rho')=\rho: P_{l'} \rightarrow \mathds{R}^+$. First, we assign ranges a follows:
\[ 
	\rho(p_j)=
		\left\{
		\begin{array}{l}
			\begin{aligned}
							c_s \cdot \rho(p_j)+ c_k \cdot T(p_j), &\: 1 \leq j \leq l, \\
							 c_k \cdot T(p_j), &\: l < j \leq m , 
				\end{aligned}
			\end{array}
		\right.		 	
\]
\end{description}

where $c_k = 1+8(1+c_s) = 19$.
The multiplicity (by $c_s$) handles the gaps caused by points breaking the \emph{\orderedline\ } condition with respect to the Euclidean metric. 
The role of the additive part, together with
the second stage of the transformation, elaborated next, is to overcome the absence of points outside the path.
In the second stage, 
for every $p_j$ with $1 \leq j \leq m$, let $1\leq j^- < j$ be the minimal index
for which there exists $u \in T(p_{j^-})$ with $|p_j u| \leq c_k\cdot w(T(p_j))$,
and let $j< j^+ \leq m$ be the maximal index
for which there exists $u \in T(p_{j^+})$ with $|p_j u| \leq c_k\cdot w(T(p_j))$,
direct the sub-path between $p_{j^-}$ and $p_{j^+}$ towards $p_j$.
See Fig.~\ref{fig:transformation} for illustration.
\begin{figure}[htb]
    \centering
        \includegraphics[width=0.78\textwidth]{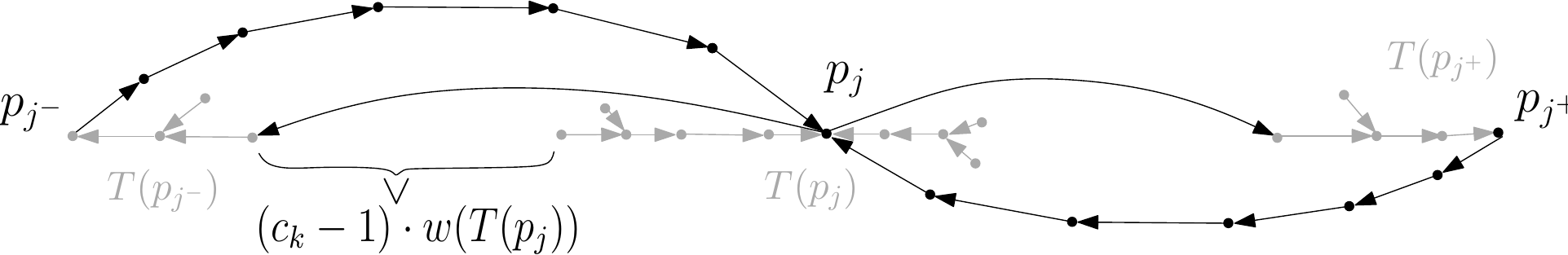}
    \caption{An illustration of the second stage in the adjustment transformation $g$.}
    \label{fig:transformation}
\end{figure}

The indexing of points in $P_M$ and notations introduced in the algorithm are used throughout this section.
The tree $T(p_i) \in R$, for $1 \leq i \leq z$, is sometimes denoted by $T_i$ for short.
%
%
\subsection{The Validity of the Output}\label{sec:validity}
We consider each solution separately and show it forms a \emph{valid} assignment $\rho$.

\noindent
\textbf{Validity of Solution (i).} \
Follows from the validity of the \emph{Hub} algorithm

\noindent
\textbf{Validity of Solution (ii).} \ 
The subgraph of $G_{\rho}$ induced by the points of $P_M$ is strongly connected due to the validity of the \emph{Hub} algorithm.
All trees in $R$ are bi-directed trees sharing a common point with $P_M$
and therefore, the whole graph, $G_{\rho}$, is strongly connected. 

\noindent
\textbf{Validity of Solution (iii).} \
Each tree in $R$ induces a strongly connected subgraph of $G_{\rho}$,
thus, it is suffices to show the connectivity of the minor obtained from $G_{\rho}$
by contracting all trees of $R$.
The subgraph of $G_{\rho}$ induced by the points of the middle sub-path $P_x$ 
forms either one directed cycle or two directed cycles sharing common vertices. 
By the correctness of the \emph{Hub} algorithm, each of the sub-paths $P_l$ and $P_r$ induces 
a strongly connected subgraph of $G_{\rho}$.
In addition, each of them shares a common vertex with the middle sub-path,
thus, the whole graph, $G_{\rho}$, is strongly connected.

\noindent
\textbf{Validity of Solution (iv).} \
Since we already verified the validity of Solution (iii), 
we are only left to show that each of $P_l$ and $P_r$ induce strongly connected subgraphs in $G_{\rho}$.
We consider the sub-path $P_l$, while the case of $P_r$ is symmetric.

Let $\rho_{l}: P_l\rightarrow \mathds{R}^+$ denote the assignments obtained by applying 
the $1$D algorithm on the sub-path $P_l$ with respect to $h_S$.
Due to the validity of $\rho_{l}$, the graph induced by $\rho_{l}$ with respect to $h_S$
is strongly connected.
Let $(p_i,p_j)$ be an edge in this graph, we show that there exists a directed path from $p_i$ to $p_j$ in $G_{\rho}$. 
Assume w.l.o.g., $1 \leq i < j \leq l$.
By the definition of $h_S$, there exist $u \in T(p_{i'})$ and $v \in T(p_{j'})$
with $1 \leq i' \leq i < j \leq j'\leq m$, such that $\rho_l (p_i) \geq h_S (p_i,p_j)=|uv|$.
Since the final assignment $\rho$ is obtained after applying the adjustment transformation $g$,
we have $\rho(p_i) \geq c_s |uv|$.
\begin{description}[topsep=0.2cm, itemsep=0.1cm]
\item[Case 1:] $|uv| \leq (c_k-1)\cdot w(T(p_{i'}))$ or $|uv| \leq (c_k-1)\cdot w(T(p_{j'}))$.\\
							 Assume, w.l.o.g., that the second condition holds, then $|p_{j'}u| \leq c_k \cdot w(T(p_{j'}))$ 
							 and thus, by the definition of transformation $g$, 
							 the directed path from $p_{i'}$ to $p_{j'}$ and the edge $(p_{j'}, u)$ are contained in $G_{\rho}$.
							 Together with the directed path in $T(p_{i'})$ from $u$ to the root $p_{i'}$ they form a cycle containing both $p_{i}$ and $p_{j}$.
\item[Case 2:] both $|uv| > (c_k-1)\cdot w(T(p_{i'}))$ and $|uv| > (c_k-1)\cdot w(T(p_{j'}))$.\\
							 Since $p_{i'},p_{j'} \in P_{l'}$, and $P_{l'}$ is the output sub-path after performing
							 the \emph{flatten} procedure, then $\delta_{P_{l'}}(p_{i'},p_{j'}) \leq c_s|p_{i'} p_{j'}|$.
							 Therefore, we have
							 \begin{align*}
							 |p_i p_j| \ &\leq \ \delta_{P_{l'}}(p_i, p_j) 
							 					 \	= \ \delta_{P_{l'}}(p_{i'},p_{j'})-(\delta_{P_{l'}}(p_{i'},p_{i})+\delta_{P_{l'}}(p_{j'},p_{j}))\\
							 						&\leq c_s|p_{i'} p_{j'}|-(|p_i p_j|-(w(T_{i'})+w(T_{j'})+|u v|))\\
							 						&\leq c_s(w(T_{i'})+w(T_{j'})+|uv|)-|p_i p_j|+w(T_{i'})+w(T_{j'})+|u v|\\
							 						&\leq (c_s+1)(w(T_{i'})+w(T_{j'}))+(c_s+1)|uv|-|p_i p_j|\\
							 						&\leq (c_s+1)(2|uv|/(c_k-1))+(c_s+1)|uv|-|p_i p_j|\\
							 						&\leq (c_s+1)(2|uv|/(8(1+c_s)))+(c_s+1)|uv|-|p_i p_j|\\
						\Rightarrow |p_i p_j| &\leq (\frac{1+c_s}{2}+\frac{1}{8})|uv| = c_s|uv|.
							 \end{align*}
					
\end{description}


\subsection{The Approximation Ratio}

Let $SOL$ denote the cost of the output of the algorithm for the input set $S$
and let $\rho^*: S \rightarrow \mathds{R}^+$ denote an optimal assignment for $S$ of cost $OPT$.
We show that $SOL \leq (1.5-\epsilon)OPT$.

Let $W=w(MST(S))$, $r = w(R)/W$, $e_{M}$ denote the longest edge in $MST(S)$
and $l = w(e_{M})/W$.
As shown in~\cite{Ambuhl05}, $OPT \geq W+w(e_M)=W(1+l)$.
Next we show several upper bounds on the ratio $SOL/OPT$, 
corresponding to the four solutions computed during the algorithms and finally conclude that the minimum
among them equals at most $(1.5-\epsilon)$.\\[2mm]
%
%
\noindent
\textbf{Approximation bound for Solution (i).} \
Due to the analysis of the \emph{Hub} algorithm done in~\cite{Ambuhl05},
we have $SOL \leq W+(w(P_{M})+w(e_{M}))/2 = W(1.5-r/2+l/2)$.
Therefore, 
\begin{align}\label{eq:i}
\frac{SOL}{OPT} \leq \frac{W(1.5-r/2+l/2)}{W(1+l)} = \frac{1.5-(r-l)/2}{1+l}.
\end{align}
Assume $SOL > (1.5-\epsilon)OPT$, then 
\begin{align*}
1.5-\epsilon < \frac{SOL}{OPT} < \frac{1.5-(r-l)/2}{1+l} < 1.5-r/2,  
\end{align*}
implies $r < 2\epsilon$ and the following corollary follows.
\begin{corollary}\label{cor:Soli}
One of the following holds, $SOL \leq (1.5-\epsilon)OPT$ or $r < 2\epsilon$.
\end{corollary}

The following lemma is crucial for introducing the bounds for Solutions (ii) and (iv).
\begin{lemma}\label{lem:sf}
Let $c_s$ and the notation $Q_{P}$ be defined as in procedure $f$. Given a constant $\delta$,  
\begin{enumerate}
\item there exist two pairs of points $(u,v),(y,w)$ connecting two disjoint sub-paths in $P_M$, each pair with stretch factor greater than $c_s$, 
		 that satisfy $\delta_{P_M}(u,v) \geq \delta$ and $\delta_{P_M}(y,w) \geq \delta$; or
\item there exists an edge $e \in {P_M}$ defining $P_{e^l}$ and $P_{e^r}$ (the two paths of $P_M\bs e$, to the left and to the right of $e$, respectively),
			such that for every $(u,v)\in Q_{P_{e^l}}\cup Q_{P_{e^r}}$,  $\delta_{P_M}(u,v)<\delta$.
\end{enumerate}
\end{lemma}
\begin{proof}
If the first condition holds, we are done. Otherwise, fix $e$ to be the rightmost edge in $P_M$. 
If the second condition does not hold for $e$, then there exists exactly one pair $(u,v)\in Q_{P_{e^l}}$
with $\delta_{P_M}(u,v)\geq \delta$. Replace $e$ with the consecutive edge to its left in ${P_M}$.
Continue the process until for every $(u,v)\in Q_{P_{e^l}}$,  $\delta_{P_M}(u,v)<\delta$. 
Note that this condition hold when $P_{e^l}$ contains a single edge.
If the process ends with a separating edge $e$ for which there exists a pair $(u,v)\in Q_{P_{e^r}}$
with $\delta_{P_M}(u,v)\geq \delta$, then $u$ is a common endpoint with the preceding edge in the process and there exists a point $w\in P_{e^l}$
such that the two pairs $(w,u),(u,v)$ satisfy the first condition.
\end{proof}

%
\noindent
\textbf{Approximation bound for Solution (ii).} \
Consider the value $c_h W$, where $c_h=20 \epsilon$. 
One of the two conditions of Lemma~\ref{lem:sf}, denoted by L\ref{lem:sf}.1 and L\ref{lem:sf}.2, respectively, must hold for $\delta = c_h W$.
We start by assuming that condition L\ref{lem:sf}.1 holds which leads to Lemma~\ref{lem:approx1}.

For every $q \in P_M$, we have that the cost of Solution (ii) 
equals at most $w(P_M)+\max\{|q p_1|,|q p_z|\}+2w(R)$. 
Consider the path $\sim$$f(P_M)$ obtained from $P_M$ after applying the \emph{flatten} procedure $f$ on the sub-paths $P_{e^l}$ and $P_{e^r}$.
Note that the length of this path is at most $w(P_M)-2c_h W(1-1/c_s)$ and it shares common endpoints with $P_M$. 
Let $\tilde{c}$ be the point on $\sim$$f(P_M)$ closest to its midpoint. 
The midpoint may lie on an edge of $P_M$ or on a shortcut performed by $f$.
If it lies on a shortcut, we undo it and only one shortcut remains.
Thus, the point $\tilde{c}$ is at Euclidean distance at most \
\(  \frac{1}{2}[ w(P_M)-c_h W(1-\frac{1}{c_s})+w(e_M) ]    \), \
from both endpoints and we have
\[ SOL \leq  W(1-r+ \frac{1}{2}[(1-r)-c_h(1-\frac{1}{c_s})+l] +2r)
		 =  	W[1.5-\frac{1}{2}(c_h(1-\frac{1}{c_s})-l-r)]. \]
%
%
This implies, 
\begin{align}\label{eq:ii}
\frac{SOL}{OPT} \leq \frac{1.5-\frac{1}{2}(c_h(1-\frac{1}{c_s})-l-r)}{1+l}.
\end{align}

\begin{lemma}\label{lem:approx1}
If condition L\ref{lem:sf}.1 holds for $\delta = c_h W$
then $SOL \leq (1.5-\epsilon)OPT$. 
\end{lemma}
\begin{proof}
Assume towards contradiction that $SOL > (1.5-\epsilon)OPT$, then
by Corollary~\ref{cor:Soli}, $r<2\epsilon$ and together with equation~(\ref{eq:ii}) we receive \\ 
%
\begin{align*}
1.5-\epsilon & < \frac{1.5-\frac{1}{2}(c_h(1-\frac{1}{c_s})-l-r)}{1+l} \ < \ 1.5-\frac{1}{2}(c_h(1-\frac{1}{c_s})-r) \\
						 & <	1.5-\frac{1}{2}(20\epsilon(\frac{1}{5})-2\epsilon) \ <	\ 1.5-\epsilon \qquad  \qquad \Rightarrow\ \text{ \ \ contradiction.}						 
\end{align*}
\end{proof}
From now on, assume condition L\ref{lem:sf}.2 holds for $\delta = c_h W$
and let $e \in P_M$ be the edge satisfying the condition.
Let $t = (\sum_{(u,v)\in Q_{P_{e^l}}}\delta_{P_M}(u,v) + \sum_{(u,v)\in Q_{P_{e^r}}}\delta_{P_M}(u,v))\ / \ W$,
%
we give an additional bound to the cost of Solution (ii)
using the same arguments used for the case where condition L\ref{lem:sf}.1 holds. 
Let $\tilde{c}$ be the point on $\sim$$f(P_M)$ closest to its midpoint.
Since every pair $(u,v)\in Q_{P_{e^l}}\cup Q_{P_{e^r}}$ satisfies $\delta_{P_M}(u,v) \leq c_h W$,
the point $\tilde{c}$ is at Euclidean distance at most \ \( \frac{1}{2} [ w(P_M)-(t-c_h) \cdot W(1-\frac{1}{c_s})+w(e_M)] \)  \
%
from both endpoints, thus 
\begin{align}\label{eq:iii}
\frac{SOL}{OPT} \leq \frac{1.5-\frac{1}{2}((t-c_h)(1-\frac{1}{c_s})-l-r)}{1+l}.
\end{align}
Note that although the analysis for equation~\ref{eq:iii} considers the path $\sim$$f(P_M)$,
the \emph{flatten} procedure $f$ is performed after Solution (ii) is computed.

In the analysis of Solution (iii) and (iv) we have $w(R)\leq(r+t)W$, since 
applying $f$ on $P_{e^l}$ and $P_{e^r}$ moves portions of the paths connecting points in $Q_{P_{e^l}}\cup Q_{P_{e^r}}$ to $R$.
Consider the iteration of the algorithm for the edge $e$ and a choice of $4$ points $p_{l}, p_{l'}, p_{r'}, p_{r} \in P_M$ satisfying:
$p_{r}$ is the rightmost point in $P_M$ with $u_r \in T(p_r)$ 
connected (at any direction) in $G_{\rho^*}$ to a point in $T(p_{l'})$ for $p_{l'}$ to the left of $e$ 
and, symmetrically, $p_{l}$ is the leftmost point in $P_M$ with $u_l \in T(p_l)$ 
connected (at any direction) to a point in $T(p_{r'})$ for $p_{r'}$ to the right of $e$.
Meaning, there is no edge in $G_{\rho^*}$ connecting between a point in $T(p)$ for $p \in P_l\bs \{p_l\}$ and a point in $T(q)$  for $q \in P_{r'}$
and no edge connecting between a point in $T(p)$ for $p \in P_r\bs \{p_r\}$ and a point in $T(q)$ for $q \in P_{l'}$.
Let $x$ denote the ratio $ w(P_x)  /  W. $
\\[2mm]
%
%
\noindent
\textbf{Approximation bound for Solution (iii).} \
%
Preforming the \emph{Hub} algorithm on $P_l$ and $P_r$, separately, result
in two assignments with a total cost of at most $1.5(W-w(P_x)-w(R)) + w(e_M)/2 + w(e_M)/2 $.
Directing the path $P_x$ and the trees in $R$, and assigning all roots in $R$ their assignment, 
together with adding the two edges, $(u_l, u_{r'})$ for $u_l \in T(p_l), u_{r'} \in T(p_{r'})$ of minimal length
and $(u_{l'},u_r)$ for $u_{l'} \in T(p_{l'}), u_r \in T(p_r)$ of minimal length 
(or the edge $(p_l,p_r)$ instead if it is cheaper)
costs at most \
\(w(P_x)+ 2\cdot w(R) + |u_l u_{r'}|+|u_{l'} u_{r}|. \)
Overall, we have a total cost of at most \
\( W + \frac{1}{2}(1-x + (r+t)) \cdot W + l W + |u_l u_{r'}|+|u_{l'} u_{r}|.\)

Since there is an edge connecting a point in $T(p_l)$ with a point in $T(p_{r'})$ (in some direction)
and an edge connecting a point in $T(p_{l'})$ with a point in $T(p_r)$ in
the optimal solution, we have $OPT \geq W + |u_l u_{r'}|+|u_{l'} u_{r}| - w(e_M)$,
hence,
\begin{align}\label{eq:iv}
\frac{SOL}{OPT} & \leq \frac{ W(1-l+ \frac{1}{2}(1-x + (r+t)) + 2 l) +|u_l u_{r'}|+|u_{l'} u_{r}| }
														{W(1-l) + |u_l u_{r'}|+|u_{l'} u_{r}|} \nonumber\\
								& \leq 1+ \frac{W(\frac{1}{2}(1-x + (r+t)) + 2l)}
															 {W(1-l) + |u_l u_{r'}|+|u_{l'} u_{r}|} \leq 1+ \frac{1}{2}(1-x + (r+t)) + 2l.
\end{align}
%
%

\noindent
\textbf{Approximation bound for Solution (iv).} \
Let $\rho_{l}: P_l\rightarrow \mathds{R}^+$ and $\rho_{r}: P_r\rightarrow \mathds{R}^+$ 
denote the assignments obtained by applying the \emph{$1$D RA} algorithm 
on $P_l$ and $P_r$, respectively, with respect to the distance function $h_S$. 
Let $\rho': P_l\cup P_r \rightarrow \mathds{R}^+$ denote the union of the two assignments
and let $OPT'$ denote the cost of $\rho'$, i.e., $OPT' = \sum_{v \in P_l \cup P_r} \rho'(v)$.
%
\begin{claim}\label{cl:OPT'}
$OPT' \leq OPT$.
\end{claim}
\begin{proof}
We show that the optimal assignment $\rho^*$ can be adjusted to an assignment\\
\mbox{$\rho:P_l \cup P_r \rightarrow \mathds{R}^+$}, valid for $P_l$ and $P_r$, separately, with respect to $h_S$, of the same cost.\\
We define,
$$\rho(p_l)=\max_{ \substack{v\in T(p_i),\\ l\leq i \leq m}} \{\rho^*(v)\}, \qquad 
\rho(p_r)=\max_{ \substack{v\in T(p_i),\\ m+1\leq i \leq r}} \{\rho^*(v)\},$$
and for every $p_j$ with $j \in \{1,..,l-1\} \cup \{r+1,..,z\}$, \ \  $\rho(p_j)=\max_{v\in T(p_j)} \{ \rho^*(v)\}.$

Let $u$ and $v$ be two points in the same sub-path, w.l.o.g., $P_l$, 
and let $(u=u_1,u_2,...,u_k=v)$ be the path from $u$ to $v$ in $G_{\rho^*}$.
Consider the sequence $(u=y_1,y_2,...,y_k=v)$, obtained by replacing every $u_i \in T(P_l)$ with $y_i=r(u_i)$,
and every $u_i \in T(P_x\cup P_r)$ with $y_i=p_l$, for $1 \leq i \leq k$.
We prove the above sequence forms a path from $u$ to $v$ 
in the graph induced by $\rho$ with respect to $P_l$ and $h_S$, and conclude that $\rho$ is valid and $OPT' \leq OPT$.
Consider a pair of consecutive nodes in the above sequence, $(y_i,y_{i+1})$. 
Note that if $u_i=p_j$ (resp. $u_{i+1}=p_j$) for $1 \leq j < l$, than $u_{i+1}=p_h$ (resp. $u_{i}=p_h$) for $1 \leq h \leq m$; thus,
for every $1 \leq i < k$, by the definition of $\rho$ and $h_S$ we have, 
$\rho(y_i)\geq \rho^*(u_i) \geq |u_i u_{i+1}|\geq h_S(y_i, y_{i+1})$ 
and the claim follows.
\end{proof}

Let $\sim$$g(\rho'):P_{l'}\cup P_{r'} \rightarrow \mathds{R}^+$ be the union of the assignments obtained
after applying the transformation $g$ on each of $P_l$ and $P_r$, separately.
The following lemma bounds its cost.
\begin{lemma}\label{lem:g(rho)}
$cost(\sim$$g(\rho')) \leq [c_s + (c_k + 2c_s(c_k +1))(r+t)]OPT.$
\end{lemma}
\begin{proof}
Consider applying transformation $g$ on $P_{l}$ and $P_{r}$.
By multiplying the range of every point in $P_{l}\cup P_{r}$ by $c_s$
we obtain an assignment of cost $c_s\cdot OPT'$.
As for the additive part and the second stage, we analyze the cost with respect to $P_{l'}$, while the case for $P_{r'}$ is symmetric.
Every $p_j \in P_{l'}$, is responsible for an additional of at most $X_j =
 c_k \cdot w(T_j)  + \delta_{P_{l'}}(p_{j^-},p_{j^+}) \leq c_k \cdot w(T_j) + c_s|p_{j^-} p_{j^+}|  =
 c_k \cdot w(T_j)  + c_s[2c_k \cdot w(T_j) + w(T_{j^-}) + w(T_{j^+}) ]$
to the total cost,
where $p_{j^-}$ and $p_{j^+}$ are defined as in the definition of $g$.
The first element in the summation is the range added to $p_j$ itself,
and the second is the cost of directing the path between $p_{j^-}$
and $p_{j^+}$ towards $p_j$ (depicted in Fig.~\ref{fig:transformation} in black).

Allegedly, for computing $cost(\sim$$g(\rho'))$ we should sum $X_j$ over all $p_j \in P_{l'} \cup P_{r'}$ and add it to the cost $c_s\cdot OPT'$,
however, we observe that it is suffices to consider a point $p_i$ only once
as $p_{j^-}$, for the rightmost point $p_j$ such that $i=j^-$ and only once
as $p_{j^+}$, for the leftmost point $p_k$ such that $i=k^+$.
Thus, we can charge $p_{j^-}$ and $p_{j^+}$ themselves once on 
each of the elements $c_s \cdot w(T_{j^-})$ and $c_s \cdot w(T_{j^+})$ in the overall summation.
Namely, charge every point $p_j$ for a total range increase of $Y_j=w(T_j)[c_k  + c_s(2c_k +2)]$.
Summing $Y_j$ over all $p_j \in P_{l'}\cup P_{r'}$ and adding the cost $c_s \cdot OPT'$, using Claim~\ref{cl:OPT'} gives 
\begin{align*}
cost(g(rho')) & \leq c_s \cdot OPT' + \sum_{1 \leq j \leq z} w(T_j)[c_k  + c_s(2c_k +2)]\\
							& = c_s \cdot OPT' + [c_k + 2c_s(c_k +1)](w(R))\\
							& \leq [c_s + (c_k + 2c_s(c_k +1))(r+t)]OPT.
\end{align*}
\end{proof}

Note that $g$ has already assigned to every $p_j \in P_{l'}\cup P_{r'}$ an assignment greater than $w(T(p_j))$.
Directing all trees in $R$ towards their roots,
directing the path $P_x$ and adding the edge $(p_l,p_r)$, adds to the cost
at most $(2x+(r+t)) W < (2x+(r+t)) OPT$, and
together with Lemma~\ref{lem:g(rho)} we receive
\begin{align}\label{eq:v}
\frac{SOL}{OPT} & \leq c_s + (c_k + 2c_s(c_k +1)+1)(r+t)+2x
\end{align}
\begin{lemma}\label{lem:approx2}
If condition L\ref{lem:sf}.2 holds for $\delta = c_h W$,
then $SOL \leq (1.5-\epsilon)OPT$.
\end{lemma}
\begin{proof}
Assume towards contradiction that $SOL > (1.5-\epsilon)OPT$. 
By Corollary~\ref{cor:Soli}, $r<2\epsilon$ and together with equation~(\ref{eq:iii}) we receive,
\begin{align*}
1.5-\epsilon  &< \frac{1.5-\frac{1}{2}((t-c_h)(1-\frac{1}{c_s})-l-r)}{1+l}  < 1.5-\frac{1}{2}((t-20\epsilon)\frac{1}{5}-r)  \\
							&< 1.5-\frac{t}{10}+3\epsilon \ \ \qquad \qquad  \Rightarrow   \ \ \ t  < 40\epsilon. 
\end{align*}
Replacing $r$ and $t$ with the above upper bounds in equation~(\ref{eq:v}) gives,
\begin{align*}
						1.5-\epsilon &< c_s + (c_k + 2c_s(c_k +1)+1)(r+t)+2x \\
							 					 &< 1\frac{1}{4}+70(42 \epsilon)+2x \ \ \qquad \Rightarrow \ \ \ x > \frac{1}{8}-1472\epsilon,
\end{align*}
and by equation~(\ref{eq:iv}) we have,
\begin{align*}
1.5-\epsilon &< 1+ \frac{1}{2}(1-x + (r+t)) + 2l < 1+ \frac{1}{2}(1-(\frac{1}{8}-1472\epsilon)+ 42\epsilon) + 2l \\
						 &< 1.5 - \frac{1}{16}+757\epsilon+2l  \ \ \ \ \Rightarrow  \ \ \  l > \frac{1}{32}-380\epsilon.							 
\end{align*}
The upper bound on $l$ together with equation~(\ref{eq:i}) imply,
\begin{align*}
1.5-\epsilon  &< \frac{1.5-(r-l)/2}{1+l}  < \frac{1}{2} + \frac{1}{1+l} 
						  < \frac{1}{2} + \frac{1}{1+\frac{1}{32}-380\epsilon} \qquad  \Rightarrow  \ \ \ \epsilon > \frac{8}{10^5}
\end{align*}
in contradiction to our choice of $\epsilon=\frac{5}{10^5}$.
\end{proof}
We conclude with Theorem~\ref{theo:2Dmain}, derived from
Lemma~\ref{lem:sf} together with Lemmas~\ref{lem:approx1} and~\ref{lem:approx2}.

\begin{theorem}\label{theo:2Dmain}
Given a set $S$ of points in $\mathds{R}^d$ for $d \geq 2$,
a minimum cost range assignment $(1.5-\epsilon)$-approximation can be computed in polynomial time for $S$, where $\epsilon=\frac{5}{10^5}$. 
\end{theorem} 

The reader can notice that our algorithm yields a better approximation bound than stated in the above theorem.
However, we preferred the simplicity of presentation over a more complicated analysis resulting in a tighter bound.

\bibliographystyle{plain}

\end{document}